\documentclass[11pt,a4paper]{article}
\usepackage[utf8]{inputenc}
\usepackage[english]{babel}
\usepackage{graphicx}
\usepackage{hyperref} 
\usepackage{geometry}
\usepackage{amsmath}
\usepackage{amssymb}
\usepackage{amsfonts}
\usepackage{multicol}
\usepackage{float}
\usepackage{caption}
\usepackage{subcaption}
\usepackage{cite}
\numberwithin{equation}{section}
\usepackage{amsthm}
\usepackage[utf8]{inputenc}
\usepackage{xcolor}

\newtheorem{theorem}{Theorem}
\newtheorem{lemma}{Lemma}

\newcommand{\gn}{G_{\rm N}}
\newcommand{\Sm}{S_{\rm m}}

\newcommand{\Lm}{\mathcal{L}_{\rm m}}

\newcommand{\Lop}{\hat{L}}
\newcommand{\Ltop}{\hat{\tilde{L}}}

\title{\bf Singularity theorems in the effective field theory for quantum gravity at second order in curvature}
\author{Folkert~Kuipers\thanks{E-mail: F.Kuipers@sussex.ac.uk}~ and~ Xavier~Calmet\thanks{E-mail: X.Calmet@sussex.ac.uk}
	\\
	{\em Department of Physics and Astronomy, University of Sussex,}\\ 
	{\em Brighton, BN1 9QH, United Kingdom}
}

\begin{document}
	
\maketitle
\vspace{5cm}

\begin{abstract}
In this paper we discuss singularity theorems in quantum gravity using effective field theory methods. To second order in curvature, this effective field theory contains two new degrees of freedom which have important implications for the derivation of these energy theorems: a massive spin-2 field and a massive spin-0 field. Using an explicit mapping of this theory from the Jordan frame to the Einstein frame, we show that both the cosmological and the black hole singularity theorems may not hold due to the presence of a massive spin-2 field in the particle spectrum of quantum gravity. Furthermore, we show that the massive scalar field can lead to a violation of the assumptions used to derive Hawking's singularity theorem. On the other hand, it does not affect Penrose's singularity theorem.
\end{abstract}

\thispagestyle{empty}
\pagebreak
\pagenumbering{arabic}

\section{Introduction}
The significance of singularity theorems in general relativity first presented in the seminal papers of Penrose and Hawking \cite{Penrose:1964wq,Hawking:1966vg} cannot be overemphasized. Since these foundational works several adaptions and refinements of the singularity theorems have been developed (see e.g. \cite{Borde:1987qr,Roman:1988vv,Wald:1991xn,Fewster:2010gm,Brown:2018hym}). In general, all these theorems boil down to the same principle: the assumption of some energy condition together with some global statement about space-time leads to the prediction of geodesic incompleteness somewhere in the space-time. Geodesic incompleteness is then often taken as equivalent to the existence of a singularity, although the latter is a slightly stronger statement (see e.g. \cite{Witten:2019qhl}).\\

A crucial ingredient for the proof of most of singularity theorems is the Raychaudhuri equation\footnote{However, see \cite{Fewster:2019bjg} for a recent example that doesn't make use of this equation}, that can be derived from the Einstein field equations. It is therefore crucial to assume classical general relativity for  singularity theorems to hold, and for any deviations of general relativity one would have to reassess the derivation of singularity theorems, as was done, for example, for $f(R)$ gravity \cite{Alani:2016ail}.\\

It is clear that general relativity needs to be embedded in a gravitational theory which can be quantized, i.e. a theory of quantum gravity, if one accounts for the quantum properties of matter and space-time. Such a theory of quantum gravity is not known yet, but many different approaches to such a theory have been formulated. Furthermore any theory of quantum gravity should in the infrared limit reduce to general relativity. Despite the lack of a unique theory of quantum gravity, quantum corrections to general relativity solutions can be calculated using effective field theory methods \cite{Weinberg:1980gg,Barvinsky:1984jd,Barvinsky:1985an,Barvinsky:1987uw,Barvinsky:1990up,Buchbinder:1992rb,Donoghue:1994dn}. Calculations done in this framework apply to any ultra-violet complete theory of quantum gravity and are valid at energies scales up to the Planck mass, and thus in the entire spectrum that can potentially be probed experimentally.\\

It is expected that in a theory for quantum gravity singularities will be resolved, since singularities lead to pathologies both in general relativity and quantum field theory. However, singularities cannot be avoided as long as  singularity theorems hold. It is therefore an important question whether the assumptions of the singularity theorems break down in a theory for quantum gravity. A discussion of possible quantum loop holes for the singularity theorems can for example be found in \cite{Ford:2003qt}.\\
	
In this work we discuss the validity of the singularity theorems in the framework of the effective field theory approach to quantum gravity. A drawback of this approach is that the theory is not valid at energy scales larger than the Planck mass which corresponds to regions of large curvature, where singularities are expected to form. We shall assume that the physics responsible for the avoidance of singularities becomes relevant at energies below the Planck mass and can thus be described within our mathematical framework, an example would be, e.g., a bounce solution in FLRW cosmology which would avoid a Big Crunch solution, see for example \cite{Donoghue:2014yha}. We note that this approach goes beyond general relativity and it is applicable to any theory of quantum gravity.\\

This paper is organized as follows: in the next section we derive the action for effective quantum gravity in the Einstein frame. In section 4 we discuss singularity theorems in effective quantum gravity using this action. In section 4 we then conclude. Furthermore, in appendix A we discuss the classical Hawking and Penrose singularity theorems, and in appendix B we discuss a refined statement of Hawking's theorem using weakened energy conditions.\\

In this paper we work in the $(+---)$ metric and use the conventions $R^{\rho}_{~\sigma\mu\nu} = \partial_{\mu} \Gamma^{\rho}_{\nu\sigma} - ...$, $R_{\mu\nu}=R^{\lambda}_{~\mu\lambda\nu}$, $T_{\mu\nu} = \frac{2}{\sqrt{|g|}} \frac{\delta \Sm}{\delta g^{\mu\nu}}$. Furthermore $\kappa^2=8\pi\gn$.\\ 
\section{Effective quantum gravity in the Einstein Frame}
In this section we map the effective field theory for quantum gravity to the Einstein frame. Such mappings for $R$ and $R_{\mu\nu}$ theories have been discussed in \cite{Magnano:1987zz,Ferraris:1988zz,Magnano:1990qu,Jakubiec:1988ef,Calmet:2012eq,Calmet:2017voc}. Furthermore, the case of effective gravity without non-local interactions has been discussed in \cite{Hindawi:1995an}. Here we adapt these approaches to include the non-local terms in the effective quantum gravity formalism. The effective action for quantum gravity can be obtained by integrating out the graviton fluctuations and potentially other massless degrees of freedom. It is known that the graviton self interactions \cite{Kallosh:1978wt} make the form factors ill-defined, as the Wilson coefficients become gauge dependent. However, there is a well defined procedure to resolve these ambiguities \cite{Barvinsky:1984jd,Barvinsky:1985an}. The resulting effective action is given by
\begin{align}
	S = \int d^4 x \sqrt{|g|}
	&\left\{-\frac{R}{2\kappa^2} 
	+ c_1(\mu) R^2 
	+ c_2(\mu) R_{\mu\nu} R^{\mu\nu} 
	+ c_3(\mu) R_{\mu\nu\rho\sigma} R^{\mu\nu\rho\sigma} 
	+ \alpha R \ln\left(\frac{\Box}{\mu^2}\right) R \right. \nonumber\\
	&\left. + \beta R_{\mu\nu} \ln\left(\frac{\Box}{\mu^2}\right) R^{\mu\nu} 
	+ \gamma R_{\mu\nu\rho\sigma} \ln\left(\frac{\Box}{\mu^2}\right) R^{\mu\nu\rho\sigma} 
	+ \mathcal{O}(\kappa^2) \right\} + \Sm.
\end{align}
Using the Gauss-Bonnet theorem this can be rewritten to\footnote{Due to the presence of a $\ln(\Box)$ term in $\Lop_2$, the Gauss-Bonnet theorem does not hold in full generality. However, it is valid up to this order in $\kappa$ \cite{Calmet:2018elv,Barvinsky:1990up,Barvinsky:1994hw,Barvinsky:1993en}} 
\begin{equation}
	S = - \frac{1}{2\kappa^2} \int d^4 x \sqrt{|g|} \left\{ R - \kappa^2 R \Lop_1 R 
	- \kappa^2 C_{\mu\nu\rho\sigma} \Lop_2 C^{\mu\nu\rho\sigma} + \mathcal{O}(\kappa^4) \right\} + \Sm,
\end{equation}
where $C$ is the Weyl tensor and
\begin{align}
	\Lop_1 &= \frac{2}{3} \left[ 3 c_1(\mu) + c_2(\mu) + c_3(\mu) 
		+ (3\alpha+\beta+\gamma) \ln \left( \frac{\Box}{\mu^2}\right) \right], \label{eq:Lop1}\\
	\Lop_2 &= \left[ c_2(\mu) + 4c_3(\mu) 
		+ (\beta + 4\gamma) \ln \left( \frac{\Box}{\mu^2}\right) \right].
\end{align}
We apply a Legendre transform to the function
\begin{equation}
	f_1(R)=R - \kappa^2 R \Lop_1 R,
\end{equation}
and find
\begin{equation}
	S = - \frac{1}{2\kappa^2} \int d^4 x \sqrt{|g|} \left\{ \phi\,R - V_1(\phi)
	- \kappa^2 C_{\mu\nu\rho\sigma} \Lop_2 C^{\mu\nu\rho\sigma} + \mathcal{O}(\kappa^4) \right\} + \Sm,
\end{equation}
where
\begin{align}
	R &= \frac{\partial V_1(\phi)}{\partial \phi},\\
	\phi &= \frac{\partial f_1(R)}{\partial R}.
\end{align}
We integrate the first equation and fix the integration constant such that
\begin{equation}\label{eq:V1}
	V_1(\phi) = -\frac{1}{4\kappa^2} (\phi -1) \Lop_1^{-1} (\phi -1),
\end{equation}
where we use the notation $\Lop_1^{-1}$ to denote the Green's function of the operator $\Lop_1$. If we apply a conformal transformation to the metric
\begin{equation}
	g_{\mu\nu} \rightarrow \bar{g}_{\mu\nu} = |\phi| g_{\mu\nu} = \exp\left({\sqrt{\frac{2\kappa^2}{3}} \chi}\right) g_{\mu\nu},
\end{equation}
where we have introduced a new field $\chi$, we can rewrite the action as
\begin{align}
	S = -\frac{1}{2\kappa^2} \int d^4 x \sqrt{|\bar{g}|} 
	&\left\{ \bar{R} + \sqrt{6} \kappa \bar{\Box} \chi 
	- \kappa^2 \bar{\nabla}^{\mu} \chi \bar{\nabla}_{\mu} \chi  
	- \frac{V_1[\phi(\chi)]}{\phi(\chi)^2}
	- \kappa^2 \bar{C}_{\mu\nu\rho\sigma} \Lop_2 \bar{C}^{\mu\nu\rho\sigma} \right. \nonumber\\
	&\left. - \frac{2\kappa^2 \Lm\left(X,g^{\mu\nu}\right)}{\phi(\chi)^2} 
	+ \mathcal{O}(\kappa^4) \right\},
\end{align}
where we have used that the Weyl tensor does not transform under a conformal rescaling of the metric. Furthermore, $X$ represents all matter fields.\\

We can drop the total divergence term, since it does not affect the equations of motion, and apply the Gauss-Bonnet theorem to rewrite the Weyl tensor. We then find
\begin{align}
	S = -\frac{1}{2\kappa^2} \int d^4 x \sqrt{|\bar{g}|} 
	&\left\{ \bar{R}
	- \kappa^2 \bar{\nabla}^{\mu} \chi \bar{\nabla}_{\mu} \chi  
	- \frac{V_1[\phi(\chi)]}{\phi(\chi)^2}
	- 2 \kappa^2 \bar{R}_{\mu\nu} \Lop_2 \bar{R}^{\mu\nu}
	+ \frac{2 \kappa^2}{3} \bar{R} \Lop_2 \bar{R}
	\right. \nonumber\\
	&\left. - \frac{2\kappa^2 \Lm\left(X,g^{\mu\nu}\right)}{\phi(\chi)^2} 
	+ \mathcal{O}(\kappa^4) \right\}.
\end{align}
We consider the function
\begin{equation}
	f_2(\bar{R}_{\mu\nu}) = \bar{R} - 2 \kappa^2 \bar{R}_{\mu\nu} \Lop_2 \bar{R}^{\mu\nu}
	+ \frac{2 \kappa^2}{3} \bar{R} \Lop_2 \bar{R},
\end{equation}
and apply a Legendre transform to this part of the action, which results in
\begin{align}
	S = -\frac{1}{2\kappa^2} \int d^4 x \sqrt{|\bar{g}|} 
	&\left\{ \psi^{\mu\nu} \bar{R}_{\mu\nu} - V_2(\psi^{\mu\nu})
	- \kappa^2 \bar{\nabla}^{\mu} \chi \bar{\nabla}_{\mu} \chi  
	- \frac{V_1[\phi(\chi)]}{\phi(\chi)^2}
	\right. \nonumber\\
	&\left. - \frac{2\kappa^2 \Lm\left(X,g^{\mu\nu}\right)}{\phi(\chi)^2} 
	+ \mathcal{O}(\kappa^4) \right\},
\end{align}
where\footnote{Note that the spin-2 field is symmetric in its indices, since $R_{\mu\nu}$ is symmetric.}
\begin{align}
	\bar{R}_{\mu\nu} &= \frac{\partial V_2(\psi^{\mu\nu})}{\partial \psi^{\mu\nu}},\\
	\psi^{\mu\nu} &= \frac{\partial f_2(\bar{R}_{\mu\nu})}{\partial \bar{R}_{\mu\nu}}.
\end{align}
We integrate the first equation and fix the integration constant such that\footnote{The potential $V_2$ is real, which can easily be shown by evaluating the expression.}
\begin{equation}
	V_2(\psi^{\mu\nu}) = - \frac{1}{8\kappa^2} 
	\left(\psi_{\mu\nu} - \frac{1 \mp i \sqrt{3}}{4} \psi\, \bar{g}_{\mu\nu} \mp i \sqrt{3}\, \bar{g}_{\mu\nu}\right) 
	\Lop_2^{-1} 
	\left(\psi^{\mu\nu} - \frac{1 \mp i \sqrt{3}}{4} \psi\, \bar{g}^{\mu\nu} \mp i \sqrt{3}\, \bar{g}^{\mu\nu}\right).
\end{equation}
We perform another metric transformation such that
\begin{equation}
	\bar{g}_{\mu\nu} \rightarrow \tilde{g}_{\mu\nu} = \sqrt{|\psi|}\, \bar{g}_{\mu\rho} \left( \psi^{-1} \right)^{\rho}_{~\nu},
\end{equation}
where we define the determinants
\begin{align}
	|g| &= \det\left(g_{\mu\nu}\right),\\
	|\psi| &= \det \left( \pi^{\mu}_{~\nu}\right),
\end{align}
and we write
\begin{align}
	\tilde{\psi}^{\mu}_{~\nu} &= \psi^{\mu}_{~\nu},\\
	\tilde{\psi}^{\mu\nu} &= \tilde{\psi}^{\mu}_{~\rho} \tilde{g}^{\rho\nu},\\
	\tilde{\psi}_{\mu\nu} &= \tilde{g}_{\mu\rho} \tilde{\psi}^{\rho}_{~\nu}.
\end{align}
We obtain the transformed action
\begin{align}
	S = -\frac{1}{2\kappa^2} \int d^4 x \sqrt{|\tilde{g}|} 
	& \left\{ \tilde{R} 
	- \kappa^2 \left(\psi^{-1}\right)^{\mu}_{~\nu} \tilde{\nabla}^{\nu} \chi \tilde{\nabla}_{\mu} \chi
	\right. \nonumber\\
	&
	+ \tilde{g}^{\mu\nu} \left(\tilde{\nabla}_{\rho} Q^{\rho}_{~\mu\nu} - \tilde{\nabla}_{\nu} Q^{\rho}_{~\rho\mu} 
	+ Q^{\rho}_{~\rho\sigma} Q^{\sigma}_{~\mu\nu} - Q^{\rho}_{~\sigma\mu} Q^{\sigma}_{~\rho\nu} \right)
	 \nonumber\\
	&\left.
	- \frac{V_1[\phi(\chi)]}{\phi(\chi)^2 \sqrt{|\psi|}}
	- \frac{V_2(\psi^{\mu\nu})}{\sqrt{|\psi|}}	
	- \frac{2\kappa^2 \Lm\left(X,g^{\mu\nu}\right)}{\phi(\chi)^2 \sqrt{|\psi|}} 
	+ \mathcal{O}(\kappa^4) \right\},
\end{align}
where
\begin{equation}
	Q^{\rho}_{~\mu\nu}(\psi^{\alpha}_{~\beta}) = \frac{1}{2} \bar{g}^{\rho\sigma} (\psi^{\alpha}_{~\beta})
	\left( \tilde{\nabla}_{\mu} \bar{g}_{\nu\sigma} (\psi^{\alpha}_{~\beta})
	+ \tilde{\nabla}_{\nu} \bar{g}_{\sigma\mu} (\psi^{\alpha}_{~\beta})
	- \tilde{\nabla}_{\sigma} \bar{g}_{\mu\nu} (\psi^{\alpha}_{~\beta})\right).
\end{equation}
We again drop the total derivative terms, and we define a new spin-2 field $\xi$ such that
\begin{equation}
	\psi^{\mu}_{~\nu} = \left( 1 + \frac{\kappa}{2} \xi \right) \delta^{\mu}_{\nu} - \kappa \xi^{\mu}_{~\nu}
\end{equation}
with $\xi=\xi^{\mu}_{~\mu}$. We find
\begin{equation}\label{eq:V2}
	V_2(\psi^{\mu\nu}) = - \frac{1}{8} \left(
	\xi^{\mu}_{~\nu} \Lop_2^{-1} \xi^{\nu}_{~\mu} 
	- \xi \Lop_2^{-1} \xi \right).
\end{equation}
After this transformation the action becomes
\begin{align}
	S = \int d^4 x \sqrt{|\tilde{g}|} 
	& \left\{ - \frac{\tilde{R}}{2\kappa^2} 
	+ \frac{1}{2} \tilde{\nabla}^{\nu} \chi \tilde{\nabla}_{\mu} \chi
	+ \frac{V_1[\phi(\chi)]}{2 \kappa^2 \phi(\chi)^2 \sqrt{|\psi(\xi)|}}
	\right. \nonumber\\
	&
	-\left[\frac{1}{2} \xi\tilde{\Box}\xi
	- \frac{1}{2} \xi^{\mu\nu} \tilde{\Box} \xi_{\mu\nu}
	- \xi^{\mu\nu} \tilde{\nabla}_{\mu} \tilde{\nabla}_{\nu} \xi
	+ \xi^{\mu\nu} \tilde{\nabla}_{\rho} \tilde{\nabla}_{\nu} \xi^{\rho}_{~\mu} \right] \nonumber\\
	&\left.
	+ \frac{V_2(\psi^{\mu\nu}(\xi))}{2 \kappa^2 \sqrt{|\psi(\xi)|}}	
	+ \Lm\left(X,g^{\mu\nu}\right) 
	\right\} + \mathcal{O}(\kappa),
\end{align}
where we used that $\phi(\chi) = 1 + \mathcal{O}(\kappa)$, $\psi^{\mu}_{~\nu} = \delta^\mu_{\nu} + \mathcal{O}(\kappa)$. In addition, we expand the terms containing a potential using $\Lop=\Ltop+\mathcal{O}(\kappa)$ and find 
\begin{align}
	S =  \int d^4 x \sqrt{|\tilde{g}|} 
	& \left\{ - \frac{\tilde{R}}{2\kappa^2}
	+ \frac{1}{2} \tilde{\nabla}^{\mu} \chi \tilde{\nabla}_{\mu} \chi
	- \chi (12\kappa^2\Ltop_1)^{-1} \chi
	\right. \nonumber\\
	&
	-\left[\frac{1}{2} \xi\tilde{\Box}\xi
	- \frac{1}{2} \xi^{\mu\nu} \tilde{\Box} \xi_{\mu\nu}
	- \xi^{\mu\nu} \tilde{\nabla}_{\mu} \tilde{\nabla}_{\nu} \xi
	+ \xi^{\mu\nu} \tilde{\nabla}_{\rho} \tilde{\nabla}_{\nu} \xi^{\rho}_{~\mu} \right] \nonumber\\
	&\left.
	- \left[\xi^{\mu\nu}(16\kappa^2\Ltop_2)^{-1}\xi_{\mu\nu} - \xi(16\kappa^2\Ltop_2)^{-1}\xi \right]
	+ \Lm\left(X,g^{\mu\nu}\right)
	\right\} + \mathcal{O}(\kappa),
\end{align}
where indices on $\xi$ are raised an lowered with $\tilde{g}$. We then find the equations of motion for the scalar field:
\begin{equation}
	\tilde{\Box} \chi = - (6\kappa^2\Ltop_1)^{-1} \chi + \mathcal{O}(\kappa).
\end{equation}
We can solve the equation of motion for the Green's function $(6\kappa^2\Ltop_1)^{-1}$ by Fourier transformation:
\begin{equation}
	\int d^4k \left\{ -k^2 + \frac{1}{4\kappa^2 \left[ 3 c_1(\mu) + c_2(\mu) + c_3(\mu) 
		+ (3\alpha+\beta+\gamma) \ln \left( \frac{-k^2}{\mu^2}\right) \right]} \right\} \chi(k) = \mathcal{O}(\kappa).
\end{equation}
This results in the mass of the scalar field given by
\begin{equation}\label{eq:massScal}
	m_0^2 = \frac{1}{4\kappa^2(3\alpha+\beta+\gamma)
		W\left(- \frac{1}{4\mu^2\kappa^2(3\alpha+\beta+\gamma)}
		\exp\left[\frac{3c_1(\mu) + c_2(\mu) + c_3(\mu)}{3\alpha + \beta + \gamma} \right] 	\right)},
\end{equation}
which corresponds to earlier results (see e.g. \cite{Calmet:2018qwg}). We can do a similar analysis for the tensor field, which yields (cf. \cite{Calmet:2018qwg})
\begin{equation}
	m_2^2 = \frac{1}{2\kappa^2(\beta+4\gamma)
		W\left( -\frac{1}{2\mu^2\kappa^2(\beta+4\gamma)}
		\exp\left[\frac{c_2(\mu) + 4c_3(\mu)}{\beta + 4 \gamma} \right] \right)}.
\end{equation}
This resulting action is
\begin{align}\label{eq:ActionFinal}
	S =  \int d^4 x \sqrt{|\tilde{g}|} 
	& \left\{ -\frac{\tilde{R}}{2\kappa^2}
	+ \frac{1}{2} \tilde{\nabla}^{\mu} \chi \tilde{\nabla}_{\mu} \chi
	- \frac{1}{2} m_0^2 \chi^2
	\right. \nonumber\\
	&
	-\left[\frac{1}{2} \xi\tilde{\Box}\xi
	- \frac{1}{2} \xi^{\mu\nu} \tilde{\Box} \xi_{\mu\nu}
	- \xi^{\mu\nu} \tilde{\nabla}_{\mu} \tilde{\nabla}_{\nu} \xi
	+ \xi^{\mu\nu} \tilde{\nabla}_{\rho} \tilde{\nabla}_{\nu} \xi^{\rho}_{~\mu} \right] \nonumber\\
	&\left.
	- \frac{1}{2} m_2^2 \left[\xi^{\mu\nu}\xi_{\mu\nu} - \xi\xi \right]
	+ \Lm\left(X,g^{\mu\nu}\right)
	\right\} + \mathcal{O}(\kappa).
\end{align}
We can then find the equation of motion for the metric
\begin{align}\label{eq:FieldEq}
	\left(\tilde{R}_{\mu\nu} - \frac{1}{2} \tilde{R}\, \tilde{g}_{\mu\nu} \right)
	=& \kappa^2 \left\{ 
	\tilde{T}_{\mu\nu} 
	+ \tilde{\nabla}_{\mu} \chi \tilde{\nabla}_{\nu} \chi 
	- \frac{1}{2} \tilde{g}_{\mu\nu} \tilde{\nabla}^{\rho} \chi \tilde{\nabla}_{\rho} \chi
	+ \frac{1}{2} m_0^2 \tilde{g}_{\mu\nu} \chi^2
	\right. \nonumber\\ & \qquad
	- 2\xi_{\mu\nu} \tilde{\Box} \xi
	- \xi \tilde{\nabla}_{\mu} \tilde{\nabla}_{\nu} \xi
	+ 2 \xi_{\mu\rho} \tilde{\Box} \xi^{\rho}_{~\nu}
	+ \xi^{\rho\sigma} \tilde{\nabla}_{\mu} \tilde{\nabla}_{\nu} \xi_{\rho\sigma}
	\nonumber\\ & \qquad
	+ 2 \xi^{\rho}_{~\mu} \tilde{\nabla}_{\nu} \tilde{\nabla}_{\rho} \xi
	+ 2 \xi^{\rho}_{~\mu} \tilde{\nabla}_{\rho} \tilde{\nabla}_{\nu} \xi
	+ 2 \xi^{\rho\sigma} \tilde{\nabla}_{\rho} \tilde{\nabla}_{\sigma} \xi_{\mu\nu}
	\nonumber\\ & \qquad
	- 2 \xi^{\rho}_{~\mu} \tilde{\nabla}_{\sigma} \tilde{\nabla}_{\rho} \xi^{\sigma}_{~\nu}
	- 2 \xi^{\rho}_{~\mu} \tilde{\nabla}_{\sigma} \tilde{\nabla}_{\nu} \xi^{\sigma}_{~\rho}
	- 2 \xi^{\rho\sigma} \tilde{\nabla}_{\mu} \tilde{\nabla}_{\sigma} \xi_{\nu\rho}
	\nonumber\\ & \qquad
	+ \tilde{g}_{\mu\nu} \left[
	\frac{1}{2} \xi \tilde{\Box} \xi
	- \frac{1}{2}\xi^{\rho\sigma} \tilde{\Box} \xi_{\rho\sigma}
	- \xi^{\rho\sigma} \tilde{\nabla}_{\rho} \tilde{\nabla}_{\sigma} \xi
	+ \xi^{\rho\sigma} \tilde{\nabla}_{\lambda} \tilde{\nabla}_{\sigma} \xi^{\lambda}_{~\rho}
	\right] \nonumber\\ & \qquad \left.
	- 2 m_2^2 \left[ \xi^{\rho}_{~\mu} \xi_{\nu\rho}
	- \xi_{\mu\nu} \xi \right] 
	+ \frac{1}{2} m_2^2 \tilde{g}_{\mu\nu} \left[
	\xi^{\rho\sigma} \xi_{\rho\sigma}
	- \xi \xi \right] \right\}
	\nonumber\\ &\qquad
	+ \mathcal{O}(\kappa^3).
\end{align}
This can be rewritten in the form
\begin{align}\label{eq:EQMgRed}
	\tilde{R}_{\mu\nu}
	=& \kappa^2 \left\{ 
	\tilde{T}_{\mu\nu}
	- \frac{1}{2} \tilde{T} \tilde{g}_{\mu\nu}
	+ \tilde{\nabla}_{\mu} \chi \tilde{\nabla}_{\nu} \chi 
	- \frac{1}{2} m_0^2 \tilde{g}_{\mu\nu} \chi^2
	\right. \nonumber\\ & \qquad
	- 2\xi_{\mu\nu} \tilde{\Box} \xi
	- \xi \tilde{\nabla}_{\mu} \tilde{\nabla}_{\nu} \xi
	+ 2 \xi_{\mu\rho} \tilde{\Box} \xi^{\rho}_{~\nu}
	+ \xi^{\rho\sigma} \tilde{\nabla}_{\mu} \tilde{\nabla}_{\nu} \xi_{\rho\sigma}
	\nonumber\\ & \qquad
	+ 2 \xi^{\rho}_{~\mu} \tilde{\nabla}_{\nu} \tilde{\nabla}_{\rho} \xi
	+ 2 \xi^{\rho}_{~\mu} \tilde{\nabla}_{\rho} \tilde{\nabla}_{\nu} \xi
	+ 2 \xi^{\rho\sigma} \tilde{\nabla}_{\rho} \tilde{\nabla}_{\sigma} \xi_{\mu\nu}
	\nonumber\\ & \qquad
	- 2 \xi^{\rho}_{~\mu} \tilde{\nabla}_{\sigma} \tilde{\nabla}_{\rho} \xi^{\sigma}_{~\nu}
	- 2 \xi^{\rho}_{~\mu} \tilde{\nabla}_{\sigma} \tilde{\nabla}_{\nu} \xi^{\sigma}_{~\rho}
	- 2 \xi^{\rho\sigma} \tilde{\nabla}_{\mu} \tilde{\nabla}_{\sigma} \xi_{\nu\rho}
	\nonumber\\ & \qquad
	+ \tilde{g}_{\mu\nu} \left[
	 \xi \tilde{\Box} \xi
	- \xi^{\rho\sigma} \tilde{\Box} \xi_{\rho\sigma}
	- 2 \xi^{\rho\sigma} \tilde{\nabla}_{\rho} \tilde{\nabla}_{\sigma} \xi
	+ 2 \xi^{\rho\sigma} \tilde{\nabla}_{\lambda} \tilde{\nabla}_{\sigma} \xi^{\lambda}_{~\rho}
	\right] \nonumber\\ & \qquad \left.
	- 2 m_2^2 \left[ \xi^{\rho}_{~\mu} \xi_{\nu\rho}
	- \xi_{\mu\nu} \xi \right]
	+ \frac{1}{2} m_2^2 \tilde{g}_{\mu\nu} \left[
	\xi^{\rho\sigma} \xi_{\rho\sigma}
	- \xi \xi \right] \right\}
	+ \mathcal{O}(\kappa^3).
\end{align}

\section{Singularity theorems in effective quantum gravity}

\subsection{Massive scalar field}
It is known that a massive scalar field always satisfies the null energy condition, but can easily violate the strong condition (cf. \cite{Bekenstein:1975ww, Hawking:1973uf}). The energy momentum tensor is given by
\begin{equation}
	T_{\mu\nu} = \nabla_{\mu} \chi \nabla_{\nu}\chi - \frac{1}{2} g_{\mu\nu} \left( \nabla^{\rho}\chi \nabla_{\rho} \chi + m^2 \chi_0^2\right).
\end{equation}
Hence,
\begin{equation}
	T_{\mu\nu}v^\mu v^\nu = \left(v^{\mu}\nabla_{\mu} \chi\right)^2 \geq 0,
\end{equation}
where $v$ is an arbitrary null vector. We conclude that the null energy condition is satisfied. However,
\begin{equation}
	T_{\mu\nu} - \frac{1}{2} g_{\mu\nu}T = \nabla_{\mu} \chi \nabla_{\nu}\chi - \frac{1}{2} g_{\mu\nu} m_0^2 \chi^2
\end{equation}
which leads to
\begin{equation}\label{eq:SECscal}
	\left(T_{\mu\nu} - \frac{1}{2} g_{\mu\nu}T\right)t^\mu t^\nu = \left(t^{\mu}\nabla_{\mu} \chi\right)^2 - \frac{1}{2} m_0^2 \chi^2,
\end{equation}
where $t$ is an arbitrary normalized time-like vector. We see that this expression could be both larger and smaller to $0$. Consequently the strong energy condition does not necessarily hold. We conclude that the scalar field arising in effective quantum gravity could resolve cosmological singularities, but not black hole singularities.

\subsection{Bounds on the mass of the massive scalar field}
Using the results from appendix \ref{sec:SingThmWeak} we can derive a bound on the mass of the scalar field for which the cosmological singularity theorem still holds. First consider the action \eqref{eq:ActionFinal} containing only the massive scalar. Eq. \eqref{eq:EQMgRed} then reduces to
\begin{equation}
	\tilde{R}_{\mu\nu}
	= \kappa^2 \left\{ 
	\tilde{T}_{\mu\nu}
	- \frac{1}{2} \tilde{T} \tilde{g}_{\mu\nu}
	+ \tilde{\nabla}_{\mu} \chi \tilde{\nabla}_{\nu} \chi 
	- \frac{1}{2} m_0^2 \tilde{g}_{\mu\nu} \chi^2
	\right\}
	+ \mathcal{O}(\kappa^3).
\end{equation}
Let us consider a globally hyperbolic $4$-dimensional space-time with compact Cauchy hypersurface $S$, and assume $|\chi|<\chi_{\max}$ is bounded towards te past of $S$. Then
\begin{align}
	\int_{0}^{T} e^{-\frac{2C\tau}{n-1}} R_{\mu\nu}(\tau) \hat{\gamma}^{\mu} \hat{\gamma}^{\nu}(\tau) d\tau &\geq -\frac{1}{2} \kappa^2 m_0^2 \chi_{\max}^2 \int_{0}^{T} e^{-\frac{2C\tau}{n-1}} d\tau \nonumber\\
	&\geq -\frac{3\kappa^2}{4C} m_0^2 \chi_{\max}^2,
\end{align}
where $\hat{\gamma}$ is a normalized tangent vector to a past directed time-like geodesic and where we have used the strong energy condition in the first line. We find
\begin{equation}
	-\frac{C}{2} + \int_{-T}^{0} e^{\frac{2C\tau}{n-1}} R_{\mu\nu}(\tau) \hat{\gamma}^{\mu}(\tau) \hat{\gamma}^{\nu}(\tau) d\tau
	\geq - \frac{C}{2} - \frac{3\kappa^2}{4C} m_0^2 \chi_{\max}^2
\end{equation}
for any $C>0$. The right hand side is maximized for $C = \sqrt{\frac{3}{2}} \kappa m_0 \chi_{\max}$. By Theorem~\ref{Thm:SingHawk2} we then find that $\mathcal{M}$ is past geodesically incomplete, if
\begin{equation}
	\theta > \sqrt{\frac{3}{2}}\, \kappa\, m_0\,\chi_{\max}
\end{equation}
everywhere on $S$. Hence for
\begin{equation}
m_0  < \sqrt{\frac{2}{3}} \frac{\theta_{\min}}{\kappa\, \chi_{\max}}
\end{equation}
the singularity theorem still holds.\\

We can use the expression for the mass of the scalar \eqref{eq:massScal} to find a condition for the Wilson coefficients. Let us first ignore the nonlocal terms $\alpha,\beta,\gamma$. We then find
\begin{equation}
m_0^2 = \frac{1}{4\kappa^2\left[3c_1(\mu)+c_2(\mu)+c_3(\mu)\right]}.
\end{equation}
We thus find that the singularity theorem holds for
\begin{equation}
3c_1(\mu) + c_2(\mu) + c_3(\mu)
> \frac{3\chi_{\max}^2}{8\theta_{\min}^2},
\end{equation}
where we have assumed $3c_1(\mu) + c_2(\mu) + c_3(\mu) > 0$, as the opposite would imply that the scalar field is tachyonic. If we include the non-local contributions, we find instead
\begin{equation}
	3c_1(\mu) + c_2(\mu) + c_3(\mu) > {\rm Re}\left( \frac{3\chi_{\max}^2}{8\theta_{\min}^2} 
	+ (3\alpha+\beta+\gamma) \ln\left[- \frac{3 \mu^2 \kappa^2 \chi_{\max}^2 }{2 \theta_{\min}^2} \right]\right),
\end{equation}
where only the logarithm has a complex part that accounts for the decay width of the field \cite{Calmet:2014gya,Calmet:2015pea,Calmet:2017omb}.\\

We can make an estimate of the expansion parameter for our universe, by assuming the FLRW-metric, and by assuming that we live on a compact Cauchy hypersurface with a Hubble parameter that is constant along the surface. We find
\begin{equation}
	\theta_{\min} = \frac{1}{3} H \approx 10^{-18}\, \rm{s}^{-1},
\end{equation}
where the Hubble parameter is fixed by experiment\footnote{We take $H_0\approx 70 {\rm km\, s^{-1}\, Mpc^{-1}}$}. In addition, we require an estimate for $\chi_{\max}$, which will rely on theoretical prejudice. However, for the effective action to be consistent one would expect that both the scalar and tensor fields arising in the Einstein frame do not exceed the Planck scale. We thus make the rough estimate
\begin{equation}
	\chi_{\max} = \sqrt{\frac{c^5}{8\pi \gn \hbar }} = 10^{42}\, \rm{s}^{-1}.
\end{equation}
Hence,
\begin{equation}
	\frac{3\chi_{\max}^2}{8 \theta_{\min}^2} = 10^{121}.
\end{equation}
Furthermore, the non-local part leads to a correction given by
\begin{equation}
	(3\alpha+\beta+\gamma) \ln\left[- \frac{3 \mu^2 \kappa^2 \chi_{\max}^2 }{2 \theta_{\min}^2} \right] \approx 10^2,
\end{equation}
where we have used the known values for $\alpha,\beta,\gamma$ assuming only standard model fields \cite{Kallosh:1978wt}. Furthermore, we have set the cutoff scale $\mu\approx \kappa^{-1}$. These non-local corrections are thus negligible compared to the local contributions.
We conclude that the singularity theorem holds, if
\begin{equation}
	3c_1(\mu) + c_2(\mu) + c_3(\mu) \gtrsim 10^{121}
\end{equation}
or
\begin{equation}
	m_0 \lesssim  10^{-34}\, {\rm eV}/c^2.
\end{equation}
The singularity theorem can thus be violated for a large range of values.\\

The scalar and spin-2 particles give rise to corrections to the Newtonian potential according to the formula
\begin{equation}
	\Phi(r) = - \frac{\gn m}{r} \left( 1 + \frac{1}{3} e^{- \rm{Re}(m_0)r} - \frac{4}{3} e^{- \rm{Re}(m_2)r} \right)
\end{equation}
The E\"{o}t-Wash experiment \cite{Hoyle:2004cw} sets bounds on deviations from this potential. Assuming that the corrections do not cancel each other, both corrections should satisfy these experimental bounds, i.e. 
\begin{equation}
	m_0,m_2\geq 10^{-3} \,\rm{eV}/c^2
\end{equation}
Hence, the singularity theorem can be violated for all feasible values of the Wilson coefficients.\\

It might seem counterintuitive that tiny Wilson coefficients already lead to a breakdown of the assumptions of the singularity theorems, while large Wilson coefficients do not. In particular, since the smaller the Wilson coefficients the closer the action is to the Einstein Hilbert action. However, small Wilson coefficients lead to very massive scalar fields, which can violate the strong energy condition, as can be seen in eq. \ref{eq:SECscal}. Furthermore, the Einstein equation is a second order differential equation, while the introduction of the terms quadratic in the Ricci scalar and tensor make it a fourth order equation. As is well known solutions of differential equations are generically not stable against perturbations that change the class of the differential equation (cf. \cite{Barrow:1983rx} for a discussion of this fact in the context of general relativity).

\subsection{Spin-2 massive ghost}

Let us now turn to the massive spin-2 field. Since this field is a ghost one would expect it to violate the null energy condition. Indeed we can write the energy momentum tensor explicitly
\begin{align}
	T_{\mu\nu} 
	&= - 2\xi_{\mu\nu} \tilde{\Box} \xi
	- \xi \tilde{\nabla}_{\mu} \tilde{\nabla}_{\nu} \xi
	+ 2 \xi_{\mu\rho} \tilde{\Box} \xi^{\rho}_{~\nu}
	+ \xi^{\rho\sigma} \tilde{\nabla}_{\mu} \tilde{\nabla}_{\nu} \xi_{\rho\sigma}
	\nonumber\\ 
	& \qquad
	+ 2 \xi^{\rho}_{~\mu} \tilde{\nabla}_{\nu} \tilde{\nabla}_{\rho} \xi
	+ 2 \xi^{\rho}_{~\mu} \tilde{\nabla}_{\rho} \tilde{\nabla}_{\nu} \xi
	+ 2 \xi^{\rho\sigma} \tilde{\nabla}_{\rho} \tilde{\nabla}_{\sigma} \xi_{\mu\nu}
	\nonumber\\ 
	& \qquad
	- 2 \xi^{\rho}_{~\mu} \tilde{\nabla}_{\sigma} \tilde{\nabla}_{\rho} \xi^{\sigma}_{~\nu}
	- 2 \xi^{\rho}_{~\mu} \tilde{\nabla}_{\sigma} \tilde{\nabla}_{\nu} \xi^{\sigma}_{~\rho}
	- 2 \xi^{\rho\sigma} \tilde{\nabla}_{\mu} \tilde{\nabla}_{\sigma} \xi_{\nu\rho}
	\nonumber\\ 
	& \qquad
	+ \tilde{g}_{\mu\nu} \left[
	\frac{1}{2} \xi \tilde{\Box} \xi
	- \frac{1}{2}\xi^{\rho\sigma} \tilde{\Box} \xi_{\rho\sigma}
	- \xi^{\rho\sigma} \tilde{\nabla}_{\rho} \tilde{\nabla}_{\sigma} \xi
	+ \xi^{\rho\sigma} \tilde{\nabla}_{\lambda} \tilde{\nabla}_{\sigma} \xi^{\lambda}_{~\rho}
	\right] \nonumber\\ 
	& \qquad
	- 2 m_2^2 \left[ \xi^{\rho}_{~\mu} \xi_{\nu\rho}
	- \xi_{\mu\nu} \xi \right] 
	+ \frac{1}{2} m_2^2 \tilde{g}_{\mu\nu} \left[
	\xi^{\rho\sigma} \xi_{\rho\sigma}
	- \xi \xi \right].
\end{align}
In order to show that the field can violate the null energy condition, we construct a counterexample. We consider the special case in which the tensor field is aligned with the metric:
\begin{equation}
	\xi_{\mu\nu} = \frac{1}{4} g_{\mu\nu} \xi.
\end{equation}
This results in an energy momentum tensor given by
\begin{equation}
	T_{\mu\nu} = - \frac{1}{16} \left( g_{\mu\nu} \xi \Box \xi + \xi\nabla_{\mu}\nabla_{\nu} \xi + \xi\nabla_{\nu}\nabla_{\mu} \xi\right).
\end{equation}
Hence,
\begin{align}
	T_{\mu\nu} v^\mu v^\nu 
	&= - \frac{1}{8} \xi\nabla_{\mu}\nabla_{\nu} \xi v^\mu v^\nu \nonumber\\
	&= - \frac{1}{8} (k_\mu v^\mu \xi)^2 \nonumber\\
	&\leq 0,
\end{align}
where $v$ is an arbitrary null-like vector and where we assumed the field $\xi$ to be an eigenvector of $\nabla_\mu\nabla_\nu$ with eigenvector $k_\mu k_\nu$, as is the case if the field exhibits sinusoidal behavior with wave vector $k$.\\

Since the spin-2 field can violate the null energy condition, it can violate the strong energy condition as well. We conclude that the massive spin-2 field can resolve both kinds of singularities, since it does not satisfy any of the required energy conditions.\\

The fact that the ghost field can resolve singularities is less of a surprise, if one takes into account that the ghost field leads to a repulsive contribution to Newton's potential \cite{Calmet:2017rxl,Calmet:2019odl}, and could thus result in a effective repulsive force at small distances.

\section{Conclusion and Outlook}
It is well known that the classical singularity theorems \cite{Penrose:1964wq,Hawking:1966vg} only hold if general relativity is assumed. Quantum gravity, however, leads to deviation from general relativity, as can easily be shown using effective field theory methods. Furthermore, one of the main objectives of quantum gravity theories is to resolve singularities. In this work, we have discussed the validity of the singularity theorems in the context of an effective field theory for quantum gravity at second order in curvature.\\

We have considered singularity theorems by making an explicit mapping to the Einstein frame. It is well known that the local terms in this theory give rise to an additional scalar and tensor field at second order in curvature. We have shown that the inclusion of the nonlocal terms at this order only give rise to a shift in the mass of these fields.\\

We have then shown that the massive spin-2 ghost field can easily violate the null energy condition and thus the strong energy condition as well. Although this is expected from a ghost field, it shows that the ghost field can be useful for resolving singularities in quantum gravity.  We stress that the ghost field in effective theories for quantum gravity is not problematic, since it must be treated as a classical field in this framework \cite{Calmet:2019odl}. \\

Furthermore, we have shown that the scalar field cannot resolve black hole singularities but can for certain values of the Wilson coefficients lead to resolution of cosmological singularities. These bounds follow purely from the singularity theorems formulated for weakened energy conditions in \cite{Fewster:2010gm}. It should be noted that cosmological singularity avoidance in this framework has already been found in \cite{Donoghue:2014yha}. On the other hand, black hole solutions do not get corrected at this order \cite{Calmet:2018elv} in the effective field theory framework, which is an indication that the classical black hole singularity persists at this order in an effective theory. However, other examples of singularity resolution in various theories such as higher derivative gravity \cite{Giacchini:2018gxp,Giacchini:2018wlf}, string theory \cite{Tseytlin:1995uq} and polynomial gravity models \cite{Accioly:2016qeb} have been found. \\

It is important to notice that the breakdown of the assumptions of Hawking's and Penrose's singularity theorem does not imply the non-existence of singularities. However, it does imply that singularities can potentially be avoided, which is impossible, if the assumptions hold. In particular, in the black hole case, where the ghost field violates the conditions for the singularity theorem, it is known that there are no correction to the metric at second order in curvature. The standard general relativity singularity is still present at this order in the effective field theory. A potential resolution of singularity must come from higher order curvature terms in the action. Alternatively, one could also hope that other black hole solutions \cite{Stelle:1977ry,Lu:2015cqa,Lu:2015psa} arising at second order in curvature may not be affected by singularities and are thus the solutions which are relevant physically.\\

Furthermore, we should notice that these results only hold up to second order in curvature. Inclusion of higher orders might force us back into a regime where the singularity theorems hold or might draw us further away from this regime. The effects of these terms is not negligible, since singularities form in highly curved regions of space-time. However, it is interesting that singularities can potentially already be resolved at second order in curvature and can help guide the way to singularity resolution in ultra-violet complete theories of quantum gravity.

\section*{Acknowledgments}
This work is supported in part  by the Science and Technology Facilities Council (grant number ST/P000819/1). 
\appendix

\section{Classical singularity theorems}

\subsection{Hawking's cosmological singularity theorem}
In this appendix we state and proof Hawking's singularity theorem \cite{Hawking:1966vg}.

\begin{theorem}\label{Thm:SingHawk}
	Let $\mathcal{M}$ be a globally hyperbolic $n$-dimensional space-time with $n\geq2$ and a Cauchy surface $S$. Assume that $\exists\, C>0$ such that $\theta_{x}<-C$ $\forall\, x\in S$, where $\theta=\frac{1}{2} g^{\mu\nu} \partial_{\tau} g_{\nu\mu}$ is the expansion parameter.	Furthermore, assume that matter within this space-time satisfies the strong energy condition
	\begin{equation}
	\left( T_{\mu\nu} - \frac{1}{2} g_{\mu\nu} T \right) t^\mu t^\nu \geq 0
	\end{equation}
	for every normalized time-like vector $t^{\mu}$ everywhere in the future of the Cauchy surface $S$. Then the space-time $\mathcal{M}$ is geodesically incomplete towards the future of $S$.\\
	Moreover, if $\theta_{x}>C$ $\forall\, x\in S$ and the strong energy condition is satisfied everywhere in the past of $S$, then $\mathcal{M}$ is geodesically incomplete towards the past of $S$.
\end{theorem}

\begin{proof}
	Consider an $n$-dimensional globally hyperbolic space-time $\mathcal{M}$ with Cauchy surface $S$. Then we can find an open neighborhood $\hat{S}\supset S$ and a coordinate system on $\hat{S}$ such that the metric is given by
	\begin{equation}
	ds^2 = -dt^2 + g_{ij}(t,\vec{x}) dx^i dx^j.
	\end{equation}
	In order to proof Hawking's singularity theorem \cite{Hawking:1966vg}, we can write down the Raychaudhuri equation \cite{Raychaudhuri:1953yv}:
	\begin{equation}
	\frac{d\theta}{d\tau} = -\frac{\theta^2}{n-1}  - \sigma_{\mu\nu}\sigma^{\nu\mu} - R_{\mu\nu}t^\mu t^\nu,
	\end{equation}
	where the expansion $\theta$ and shear $\sigma_{\mu\nu}$ are given by
	\begin{align}
	\theta &= \frac{1}{2} g^{\mu\nu} \partial_{\tau} g_{\nu\mu} = \frac{\dot{V}}{V},\\
	\sigma^{\mu}_{\nu} &= \frac{1}{2} \left( g^{\mu\rho} \partial_\tau g_{\rho\nu} - \frac{1}{n-1} \delta^{\mu}_{\nu} g^{\rho\sigma} \partial_{\tau} g_{\sigma\rho} \right),
	\end{align}
	where we defined
	\begin{equation}
	V = \sqrt{\det(g)}
	\end{equation}
	and the time-derivative by $\dot{V} = \partial_{\tau} V$. Furthermore, $\theta$ and $\sigma_{\mu\nu}$ are taken along a time-like path $\gamma$ parametrized by $\tau$ with normalized tangent vectors $t^\mu$, and $\gamma(0)\in S$.\\
	
	If we use the Einstein field equation, we can rewrite the Raychaudhuri equation to
	\begin{equation}
	\frac{d\theta}{d\tau}  = - \frac{\theta^2}{n-1} - \sigma_{\mu\nu}\sigma^{\nu\mu} - \kappa^2 \left( T_{\mu\nu} - \frac{1}{2} g_{\mu\nu} T \right) t^\mu t^\nu.
	\end{equation}
	Assuming the strong energy condition
	\begin{equation}
	\left( T_{\mu\nu} - \frac{1}{2} g_{\mu\nu} T \right) t^\mu t^\nu \geq 0,
	\end{equation}
	we find
	\begin{equation}
	\frac{d\theta}{d\tau} \leq - \frac{\theta^2}{n-1}.
	\end{equation}
	Hence,
	\begin{equation}\label{eq:HawkSingIneq}
	\frac{d}{d\tau} \theta^{-1} \geq \frac{1}{n-1}.
	\end{equation}
	Assume $\exists\, C>0$ such that $\theta_{x}(0)<-C$ $\forall\, x\in S$, then we can integrate \eqref{eq:HawkSingIneq} and obtain
	\begin{equation}
	\frac{1}{\theta(\tau)} \geq \frac{\tau}{n-1} - \frac{1}{C}.
	\end{equation}
	Hence for $\tau\in\left(-\infty,\frac{n-1}{C}\right)$
	\begin{equation}
	\theta(\tau) \leq - \left(\frac{1}{C} - \frac{\tau}{n-1}\right)^{-1}.
	\end{equation}
	We can rewrite in terms of $V$ and integrate to find
	\begin{equation}
	0 \leq V(\tau) \leq V(0) \left( 1 - \frac{C \tau}{n-1} \right)^{n-1}.
	\end{equation}
	Therefore
	\begin{equation}
	\lim_{\tau \rightarrow \frac{n-1}{C}} V(\tau) = 0.
	\end{equation}
	We thus conclude that any geodesic emanating from the Cauchy surface will develop a focal point for $0<\tau\leq \frac{n-1}{C}$. Furthermore, since $S$ is a Cauchy surface and $\mathcal{M}$ is globally hyperbolic, any point $y\in\mathcal{M}$ is connected to a point $x\in S$ through a causal path of maximal proper time. We thus conclude that no geodesic $\gamma(\tau)$ can be extended to $\tau\geq\frac{n-1}{C}$. Therefore, the space-time is geodesically incomplete towards the future. This proves the future version of the theorem. The past version immediately follows by inverting the time direction in the proof.\\
\end{proof}

We conclude this subsection by mentioning an immediate result of the theorem: if there exists a Cauchy surface $S$ such that the Hubble parameter $H\geq H_0>0$ on the entire surface $S$, and the strong energy condition is expected to hold anywhere in the past of this surface, then the space-time is geodesically incomplete towards the past. More precisely no geodesic can be extended beyond $\tau = H_0^{-1}$ towards the past. To see this, we recall that the Hubble constant given by
\begin{equation}
H = \frac{\dot{a}}{a} = (n-1) \frac{\dot{V}}{V}
\end{equation}
for the FLRW-metric
\begin{equation}
ds^2 = - dt^2 + a(t)^2 d\vec{x}^2.
\end{equation}

\subsection{Penrose's black hole singularity theorem}
In this appendix we state and prove Penrose's singularity theorem \cite{Penrose:1964wq}. Here we closely follow the proof provided in \cite{Witten:2019qhl}.

\begin{theorem}\label{Thm:SingPen}
	Let $\mathcal{M}$ be a globally hyperbolic $n$-dimensional space-time with $n\geq3$ and a non-compact Cauchy surface $S$. Assume that $\mathcal{M}$ contains a compact trapped surface\footnote{A codimension 2 spacelike and achronal submanifold such that the null expansion parameter is negative everywhere on $U$ for each family of orthogonal future going null geodesics.} $U$. Furthermore, assume that matter within this space-time satisfies the null energy condition
	\begin{equation}
	T_{\mu\nu} v^\mu v^\nu \geq 0
	\end{equation}
	for every null-like vector $v^{\mu}$ everywhere in the future of the trapped surface $U$. Then the space-time $\mathcal{M}$ is null-geodesically incomplete towards the future of $U$.\\
\end{theorem}

\begin{proof}
	Consider a globally hyperbolic $n$-dimensional space-time with non-compact Cauchy surface $S$, and a compact trapped surface $U$. Then we can find an open neighborhood $\hat{U} \supset U$ and a coordinate system on $\hat{U}$ such that the metric is given by (cf. \cite{Sachs:1962zzb,Witten:2019qhl})
	\begin{equation}
	ds^2 = -2 e^q dv du + g_{AB}(dx^A + c^A dv)(dx^B + c^B dv),
	\end{equation}
	where $x^A$ is an arbitrary but fixed local coordinate system on the $(n-2)$-dimensional surface $U$. Furthermore, $q$ and $c$ are respectively a scalar and vector function of the coordinates. In this metric we can evaluate the Ricci tensor and find
	\begin{equation}
	R_{uu} = -\frac{1}{2} \partial_u \left( g^{AB} \partial_u g_{AB}\right) - \frac{1}{4} \left( g^{AC} \partial_u g_{BC}\right) \left( g^{BD} \partial_u g_{DA}\right).
	\end{equation}
	We can define the area of a bundle of orthogonal null geodesics locally by
	\begin{equation}
	A = \sqrt{\det(g_{AB})},
	\end{equation}
	which allows us to define the null expansion as
	\begin{equation}
	\theta = \frac{\dot{A}}{A} = \frac{1}{2} g^{AB} \partial_u g_{BA},
	\end{equation}
	where the dot represents a derivative with respect to $u$. Furthermore, we can define the null shear by
	\begin{equation}
	\sigma^A_B = \frac{1}{2} \left( g^{AC} \partial_u g_{CB} - \frac{1}{n-2} \delta^A_B g^{CD} \partial_u g_{DC}\right).
	\end{equation}
	We then find the null Raychaudhuri equation given by
	\begin{equation}
	\frac{d\theta}{d u} = - \frac{\theta^2}{n-2} - \sigma_{AB}\sigma^{BA} - R_{uu}.
	\end{equation}
	Furthermore, we can use the Einstein equation and the fact $g_{uu}=0$ to write
	\begin{equation}
	R_{uu} = \kappa^2 T_{uu}.
	\end{equation}
	Imposing the null energy condition results in
	\begin{equation}
	\frac{d}{d u} \theta^{-1} \geq \frac{1}{n-2}.
	\end{equation}
	Using that $U$ is a trapped surface $\exists\, C>0$ such that $\theta_x<-C\; \forall x\in U$, one can  integrate this equation in a similar way as was done in the proof of theorem \ref{Thm:SingHawk}. One obtains
	\begin{equation}
	\lim_{u \rightarrow \frac{n-2}{C}} A(u) = 0.
	\end{equation}
	Therefore, all future going null like geodesics develop a focal point for an affine distance $0<u\leq\frac{n-2}{C}$.\\
	
	Let us now assume that all null-geodesics can be extended beyond this focal point, and let us pick such a geodesic $l$ arbitrarily. Then at least a small segment of this geodesic is prompt, and lies in the lightcone $\partial J^+(U)$. Furthermore, the part of $l$ that lies in $\partial J^+(U)$ is connected, and the part beyond its first focal point cannot be in $\partial J^+(U)$, since it is not prompt. Therefore  $l \cap \partial J^+(U)$ is a finite non-empty interval, which has to be closed, since $\partial J^+(U)$ is closed in $\mathcal{M}$.\\
	
	If we take an arbitrary point $p\in \partial J^+(U)$, then this point can be reached by a null geodesic originating from $U$. This point is thus determined by the point $q\in U$, where the geodesic emanates, the value of the affine parameter $u$ measured along the geodesic and the direction (i.e. ingoing or outgoing) of the geodesic. Since $U$ is compact and since the affine parameters measured along the geodesics range over a compact interval, we find that $\partial J^+(U)$ is compact.\\
	
	However, by construction $\partial J^+(U)$ is an achronal codimension 1 submanifold of $\mathcal{M}$. Furthermore, by assumption $\mathcal{M}$ is a globally hyperbolic manifold with noncompact Cauchy hypersurface $S$, and thus does not allow for an achronal codimension 1 submanifold (see e.g. \cite{Witten:2019qhl}).	Hence, we arrive at a contradiction and conclude that at least one of the future going null geodesics orthogonal to $U$ cannot be extended beyond an affine distance $(n-2)/C$, which proves the theorem.\\
\end{proof}

\section{Singularity theorems for weakened energy conditions}\label{sec:SingThmWeak}
In this section, we state a theorem and its proof from \cite{Fewster:2010gm}. The theorem is similar to Hakwing's cosmological singularity theorem, but uses relaxed conditions on the energy momentum tensor.

\begin{theorem}\label{Thm:SingHawk2}
	Let $\mathcal{M}$ be a globally hyperbolic $n$-dimensional space-time ($n\geq2$) with a compact Cauchy surface $S$. Assume that $\exists\, C\geq0$ such that along every future directed geodesic $\gamma$ issuing orthogonally from $S$ we have
	\begin{equation}\label{eq:ThmAss}
	\liminf_{T\rightarrow \infty} \int_{0}^{T} e^{-\frac{2C\tau}{n-1}} R_{\mu\nu}(\tau) \hat{\gamma}^{\mu}(\tau) \hat{\gamma}^{\nu}(\tau) d\tau > \theta(x_0) + \frac{C}{2},
	\end{equation}
	where $x_{0}= \gamma(0)\in S$, $\theta(x_0)$ is the expansion at $x_0$, and $\hat{\gamma}(\tau)$ is a normalized time-like tangent vector of $\gamma(\tau)$. Then $\mathcal{M}$ is geodesically incomplete towards the future of $S$.\\
	Moreover, if
	\begin{equation}
	\liminf_{T\rightarrow \infty} \int_{0}^{T} e^{-\frac{2C\tau}{n-1}} R_{\mu\nu}(\tau) \hat{\gamma}^{\mu}(\tau) \hat{\gamma}^{\nu}(\tau) d\tau > - \theta(x_0) + \frac{C}{2}
	\end{equation}
	with $\gamma$ a past directed geodesic, then $\mathcal{M}$ is geodesically incomplete towards the past of $S$.
\end{theorem}

For the proof we will use the following lemma which is proved in \cite{Fewster:2010gm}.

\begin{lemma}
	Consider the initial value problem
	\begin{equation}\label{eq:lemma}
	\begin{cases}
	\dot{x}(t) = \frac{x(t)^2}{q(t)} + p(t),\\
	x(0) = x_0,
	\end{cases}
	\end{equation}
	where $q(t)$ and $p(t)$ are continuous on $[0,\infty)$ and $q(t)>0 \; \forall t\in[0,\infty)$. If
	\begin{align}
	\int_{0}^{\infty} q(t)^{-1} dt &= \infty,\label{eq:CondLemmaQ}\\
	\liminf_{T\rightarrow \infty} \int_{0}^{T} p(t) dt &> -x_0,\label{eq:CondLemmaP}
	\end{align}
	eq. \eqref{eq:lemma} has no solution on $[0,\infty)$. Moreover it implies that $\lim_{t \rightarrow t_c} x(t)\rightarrow \infty$ for $t_c\in(0,\infty)$.\\
\end{lemma}

\begin{proof}[Proof of Theorem~\ref{Thm:SingHawk2}]
	We follow the same argument as in the proof of Theorem~\ref{Thm:SingHawk} and find the Raychaudhuri equation
	\begin{equation}
	\frac{d\theta}{d\tau} = -\frac{\theta^2}{n-1}  - \sigma_{\mu\nu}\sigma^{\mu\nu} - R_{\mu\nu}t^\mu t^\nu,
	\end{equation}
	which can be rewritten to
	\begin{equation}
	\frac{dx(\tau)}{d\tau} = \frac{x(\tau)^2}{q(\tau)} + p(\tau)
	\end{equation}
	with
	\begin{align}
	x(\tau) &= -(\theta + C) e^{-\frac{2C \tau}{n-1}},\\
	p(\tau) &= \left( \frac{C^2}{n-1} + \sigma_{\mu\nu}\sigma^{\mu\nu} + R_{\mu\nu} t^{\mu} t^{\nu} \right) e^{-\frac{2C \tau}{n-1}},\\
	q(\tau) &= (n-1) e^{-\frac{2C \tau}{n-1}}.
	\end{align}
	Then $q(\tau)$ satisfies condition~\eqref{eq:CondLemmaQ}, while $p(\tau)$ satisfies condition~\eqref{eq:CondLemmaP}, if
	\begin{equation}
	\liminf_{T\rightarrow \infty} \int_{0}^{T} \left( \frac{C^2}{n-1} + \sigma_{\mu\nu}\sigma^{\mu\nu} + R_{\mu\nu} t^{\mu} t^{\nu} \right) e^{-\frac{2C \tau}{n-1}} d\tau > \theta(0) + C,
	\end{equation}
	which is satisfied, if
	\begin{equation}
	\liminf_{T\rightarrow \infty} \int_{0}^{T} e^{-\frac{2C \tau}{n-1}} R_{\mu\nu} t^{\mu} t^{\nu} d\tau > \theta(0) + \frac{C}{2}.
	\end{equation}
	By assumption~\eqref{eq:ThmAss}, this holds for all geodesics emanating from the Cauchy surface $S$. Thus $\lim_{\tau \rightarrow \tau_\gamma}x(\tau)\rightarrow\infty$ for some $\tau_\gamma\in(0,\infty)$, which immediately implies that $\lim_{\tau \rightarrow \tau_\gamma}\theta(\tau)\rightarrow -\infty$. Hence
	\begin{equation}
	\forall \gamma:[0,\infty)\rightarrow\mathcal{M}\quad {\rm with}\quad \gamma(0)\in S \quad \exists \tau_\gamma\in(0,\infty)\quad {\rm s.t.}\quad \lim_{\tau \rightarrow \tau_\gamma}V(\tau)\rightarrow 0.
	\end{equation} 
	By compactness of $S$, $\sup\{\tau_\gamma|\gamma:[0,\infty)\rightarrow\mathcal{M}, \gamma(0)\in S\}<\infty$. Furthermore, since $\mathcal{M}$ is globally hyperbolic every point $y\in J^+(S)$ can be connected through a geodesic $\gamma$ with maximal proper time.\\
	The past version can be obtained with a similar proof by inverting the direction of time.\\
\end{proof}

Let us finally note that one can derive a similar theorem for the black hole case \cite{Fewster:2010gm}.

\bibliographystyle{ieeetr}
\bibliography{mybib}

\end{document}